\documentclass{amsart}
\usepackage[title]{appendix}
\usepackage{graphicx}
\usepackage{amsfonts}
\usepackage{url}
\usepackage{amscd}
\usepackage{amssymb}

\usepackage[algosection,lined,boxed,commentsnumbered,ruled,vlined]{algorithm2e}
\usepackage[hidelinks]{hyperref}

\newtheorem{remark}{Remark}[section]

\newtheorem{theorem}{Theorem}[section]

\newtheorem{definition}[theorem]{Definition}
\newtheorem{proposition}[theorem]{Proposition}
\numberwithin{equation}{section}

\usepackage{indentfirst}
\usepackage[greek,english]{babel}
\usepackage[utf8]{inputenc}
\usepackage{color,amsmath,amssymb,graphicx}
\usepackage{tikz-cd}
\usepackage{amsmath,amssymb}
\usepackage{enumerate}
\usepackage{dsfont}
\usepackage{amsthm}
\usepackage{float}

\usepackage[algosection,lined,boxed,commentsnumbered,ruled,vlined]{algorithm2e}

\usepackage{mathtools}

\author{Marios Adamoudis}
\address{Department of Mathematics, 
Aristotle University of Thessaloniki, 54 124, Thessaloniki, Greece}
\email{aamarios@math.auth.gr}

\author{K. A. Draziotis}
\address{Department of Informatics, Aristotle University of Thessaloniki, 54 124, Thessaloniki, Greece}
\email{drazioti@csd.auth.gr}
\begin{document}
	\title[Message recovery attack to NTRU]{Message recovery attack to NTRU using a lattice independent from the public key}

\maketitle	

 \begin{abstract}
In the present paper we introduce a new attack on NTRU-HPS cryptosystem using lattice theory and Babai's Nearest Plane Algorithm. This attack generalizes the classic CVP attack on NTRU. Finally, the attack is illustrated by many examples.
	
 \end{abstract}

	\section{Introduction}
	NTRU\footnote{$N$-th degree Truncated polynomial Ring Units} cryptosystem  was proposed in 1996 by the three mathematicians, Pipher, Hoffstein and Silverman, \cite{hoffstein}. It is based on certain hard problems involving lattices and can be used as an encryption system, NTRUencrypt, as well as a digital signature, like the RSA system. Notably, NTRU seems immune to quantum attacks, whereas RSA/Diffie-Hellman are vulnerable to Shor's quantum attack \cite{shor}. It was among the seven finalists of NIST's competition for Public-Key Post-Quantum Cryptographic
	Algorithms\footnote{\url{https://csrc.nist.gov/news/2020/pqc-third-round-candidate-announcement}}, but NIST will not standardize it. However, this system remains a great choice, for instance it is implemented in openssh 9.0\footnote{\url{https://www.openssh.com/txt/release-9.0}. It is implemented in a hybrid scheme, Streamlined NTRU Prime $+$ x25519 ECDH ey exchange method.}. 
	
For preliminaries in  NTRU, see \cite{hoffstein,howgrave,andreas}. 	In the present work we shall study the original NTRU system, namely NTRU-HPS. There are also two other flavors of the original NTRU $:$ NTRU-Prime and NTRU-HRSS \cite{HRSS}. The former avoids decryption failures and rings associated to cyclotomics.  Furthermore, NTRU-Prime has two sub variants $:$ the streamlined NTRU-Prime, which is similar to NTRU-HPS, and NTRU-LPRime \cite{bernstein1, bernstein2},  but it is based on the non-cyclotomic NTRU problem. All the three flavors passed to the second round of NIST's competition. In the third round NTRU-HPS and NTRU-HRSS (which have been merged) is one of the seven finalists\footnote{\url{https://csrc.nist.gov/Projects/post-quantum-cryptography/round-3-submissions}},	and NTRU-prime is an alternate candidate for round 3. For differences between the three flavors see \cite{faq}.
		
		Before we present our contribution in section \ref{sec:previous_work}, we provide some necessary definitions for NTRU cryptosystem and some classic attacks on it.
\subsection{NTRU cryptosystem}\label{sec:ntru}
	\subsubsection{Convolution Rings}
	\begin{definition}
		Let $N$ be a positive integer. We call the ring
		$$R = {\mathbb{Z}}[x]/\langle x^N-1 \rangle$$
		the ring of convolution polynomials. If we consider the ring $\mod{q}$ i.e.
		$$R_q = {\mathbb{Z}}_q[x]/\langle x^N-1 \rangle$$
		we call it the ring of convolution polynomials $\mod{q}.$
	\end{definition}
	We set ${\mathbb{F}}[x]$ to be either $R$ or $R_q.$
	The degree of a polynomial ${\bf a}(x)$ in ${\mathbb{F}}[x]$ is $<N.$ If there is a monomial of order $
	\ge N,$ say $x^{j},$ then $x^{j}=x^{N i+ \ell} =x^{\ell}$ inside the ring ${\mathbb{F}}[x].$
	Also, instead of ${\bf a}(x)$ we can consider the vector ${\bf a}$ with coordinates  the coefficients of ${\bf a}(x).$
	Let 
	$${\bf a}(x)=a_0+a_1x+\cdots+a_{N-1}x^{N-1}$$
	and
	$${\bf b}(x)=b_0+b_1x+\cdots+b_{N-1}x^{N-1}.$$	
	We define the star multiplication $\star$ in $R	$ as follows,
	$${\bf a}(x)\star {\bf b}(x) = {\bf c}(x)=c_0+c_1x+\cdots+c_{N-1}x^{N-1},$$
	where 
	$$c_k = \sum_{i+j\equiv k\pmod{N}}a_i b_{k-j},\ 0\leq k\leq N-1,$$
	and $0\leq i,j\leq N-1.$
	If ${\bf a}(x)$ and ${\bf b}(x)$ are in $R_q,$ then we 
	compute the star product using the 
	previous formula and then, we reduce $\mod{q}.$
	\subsubsection{Circulant matrices and star multiplication}
Let ${\bf h}(x)=h_0+h_1x+\cdots +h_{N-1}x^{N-1}$ be a polynomial in $R.$ We represent polynomials of $R$ as integer vectors. For instance we represent ${\bf h}(x)$ as the vector 
		$${\bf h}=(h_0,...,h_{N-1}).$$
Consider the circulant matrix
	\begin{equation}\label{formula:circulant}
		C({\bf h})=
		\left[\begin{array}{cccc}
   h_0    &  h_1   & \dots & h_{N-1}  \\
   h_{N-1}&  h_0   & \dots & h_{N-2}  \\
    \vdots & \vdots & \ddots & \vdots   \\
   h_1    &   h_2  &  \dots &   h_0  \\
		\end{array}\right].
		\end{equation}
			Then, we can easily see that, if ${\bf a}(x)\in R,$ then we get
		$${\bf a}\star{\bf h} = [{\bf a}]C({\bf h}),$$
where $[{\bf a}]$ is the row-matrix with entries the coordinates of ${\bf a}.$
	\subsubsection{NTRU-HPS cryptosystem}
	Alice chooses public parameters $(N,p,q,d)$ with $N$ and $p$ primes, $\gcd(N,q) = \gcd(p,q) = 1.$ To get an idea of the size and form of the parameters, say $p$ is small (usually $3$) and $N,q$ are large (of the same order) and $q$ is a power of $2.$  
	The degree parameter $N$ is chosen prime due to the attack of Gentry \cite{gentry} when $N$ is composite. What is more, a prime integer $N$ maximizes the probability an element of $R_q$ be invertible in the case where $q$ is a prime power. Also, we define the set of ternary polynomials 
	${\mathcal{T}}(d_1,d_2)\subset R,$ be the polynomials of $R$ with $d_1$ entries equal to one, $d_2$ entries equal to minus one, and the remaining entries are zero.	When we write ${\mathcal{L}}_F$ for some polynomial $F\in R,$ we mean a subset of ${\mathcal{T}}(d_1,d_2)$ for some $d_1,d_2,$ relatively small with respect to $q.$
	Alice chooses her private key $({\bf f}(x),{\bf g}(x)),$ such that,	${\bf f}(x) \in {\mathcal{L}}_f$ and ${\bf g}(x) \in {\mathcal{L}}_g$ \footnote{in the original NTRU-HPS, $\mathcal{L}_f=\mathcal{T}(d+1,d)$ (or $\mathcal{L}_f=\{1+pG:G\in{\mathcal{T}}(d,d) \}$), $\mathcal{L}_g=\mathcal{T}(d,d).$}, where ${\bf f}(x)$ is invertible in $R_{q}$ and $R_{p}$. If ${\bf f}(x)$ is invertible, then the inverses are easily computed in $R_p$ and $R_q$ by using Euclidean algorithm and Hensel's Lemma. Let ${\bf F}_q(x)$ and ${\bf F}_p(x)$ be the inverses of ${\bf f}(x)$ in $R_{q}$ and $R_{p},$ respectively. Alice next computes 
	$${\bf h}(x) = {\bf F}_q(x) \star {\bf g}(x)\mod{q}.$$
	The polynomial ${\bf h}(x)$ is Alice's public key. 
The problem of distinguishing ${\bf h}(x)$ from uniform elements in $R_q$ is called {\it decision NTRU problem.} Whilst, the problem of finding the private key $({\bf f}(x),{\bf g}(x))$ is called, {\it search NTRU problem.} 	
	
	Bob's plaintext is a polynomial ${\bf m}(x) \in R$ whose coefficients are integers in the interval $[-\frac{1}{2}(p-1),\frac{1}{2}(p-1)].$ In other words, the plaintext ${\bf m}(x)$ is the centerlift of a polynomial\footnote{that is, reduction of the coefficients into the interval $(-p/2,p/2].$} in $R_{p}$. Thus, if $p=3,$ then the message is the centerlift of a ternary polynomial. Bob chooses a random ephemeral key\footnote{In NTRU-HRSS the public key is ${\bf h}(x) = (x-1)\star {\bf F}_q(x) \star {\bf g}(x)\mod{q}$ and the choice of ${\bf f},{\bf g},{\bf r},{\bf m}$ is a little different from NTRU-HPS, but again all the polynomials are short ternary polynomials.}  
	${\bf r}(x) \in {\mathcal{L}}_r$ (in the original NTRU-HPS, ${\mathcal{L}}_r={\mathcal{T}}(d,d)$) and computes 
	\begin{equation}\label{definition_of_encryption}
	{\bf e}(x) \equiv p{\bf r}(x) \star {\bf h}(x) + {\bf m}(x) \, \bmod\,  q.
	\end{equation} 
	Finally, Bob sends to Alice the ciphertext 
	${\bf e}(x)\in R_{q}.$ 
	
	To decrypt, Alice computes $${\bf a}(x) \equiv {\bf f}(x) \star {\bf e}(x) \, \bmod\,  q.$$ Then, she center lifts ${\bf a}(x)$ to an element of $R$ say ${\bf a}'(x),$  and she finally computes,
	$${\bf b}(x) \equiv {\bf F}_p(x) \star {\bf a}'(x) \, \bmod\,  p.$$ Then, ${\bf b}(x)$ is equal to the plaintext ${\bf m}(x)$ (this is true when a simple inequality between $d,q$ and $N$ is satisfied).
	\subsubsection{Background on lattices}
	
	We recall some well-known facts about lattices.
	
	Let ${{\bf{b}}_1,{\bf{b}}_2,\ldots,{\bf{b}}_n}$ be linearly independent vectors of ${\mathbb{R}}^{m}$.
	The set 
	$$\mathcal{L} = \bigg{\{} \sum_{j=1}^{n}\alpha_j{\bf{b}}_j :
	\alpha_j\in\mathbb{Z}, 1\leq j\leq n\bigg{\}}$$
	is called  a {\em lattice} and 
	the finite vector set $\mathcal{B} = \{{\bf{b}}_1,\ldots,{\bf{b}}_n\}$ is called a basis of 
	the lattice $\mathcal{L}$. 
	All the bases of $\mathcal{L}$ have the same number of elements, say $n,$ which is called
	{\em dimension} or {\em rank} of $\mathcal{L}$. If $n=m$, then	the lattice $\mathcal{L}$ is said to have {\em full rank}. 
	We denote by $M$ the $n\times m$ matrix  having as rows the vectors 
	${\bf{b}}_1,\ldots,{\bf{b}}_n$. 
	If $\mathcal{L}$ has full rank, then the {\em volume} of the lattice
	$\mathcal{L}$ is defined to be the positive number
	$|\det{M}|.$  The volume, as well as the rank, are independent of the basis $\mathcal{B}$. It is denoted 
	by $vol(\mathcal{L})$ or $\det{\mathcal{L}}$ (see also \cite{gal}).
	If ${\bf v}\in \mathbb{R}^m$, then $\|{\bf v}\|$ denotes, as usually, the
	Euclidean norm of ${\bf v}$.  Further,  we denote by  $\lambda_1(\mathcal{L})$ the least of the lengths of vectors of 	$ \mathcal{L}-\{ {\bf 0} \}$. Finally, if ${\bf t}\in {\mathbb{R}}^m$, then
	with $dist(\mathcal{L},{\bf t})$ we denote $\min\{\|{\bf v}-{\bf t}\|: {\bf v}\in \mathcal{L} \}$.
	
	There are two main problems in integer lattices. The Shortest Vector Problem (SVP) and the Closest Vector Problem (CVP) and their approximation versions.
	 SVP is defined as follows: Given a lattice $\mathcal{L},$ find a non zero vector ${\bf u}\in \mathcal{L}$ such that, for every non zero ${\bf u}'\in \mathcal{L}$
	we have:
	$$\|{\bf u}|\leq \|{\bf u}'\|.$$
	
	We define the approximate Closest Vector Problem $CVP_{\gamma_{n}}(\mathcal{L})$ (for some $\gamma_n\geq 1$) as follows:
	Given a lattice $\mathcal{L}$ and a vector 
	${\bf t}\in {\mathbb{R}}^m,$ 
	find a lattice vector ${\bf u}$ such that, for every ${\bf u}'\in \mathcal{L}$
	we have:
	$$\|{\bf u}-{\bf t}\|\leq \gamma_n \|{\bf u}'-{\bf t}\|.$$ 	
	We say that we have a CVP oracle, if we have an efficient probabilistic algorithm 
	that solves ${\rm CVP}_{\gamma_n},$ for $\gamma_n=1.$ To solve 
	${\rm{CVP}}_{\gamma_n}$, we
	usually use Babai's algorithm \cite[Chapter 18]{gal} (which has polynomial running
	time). In fact, combining this algorithm with the LLL algorithm, we solve 
	${\rm{CVP}}_{\gamma_n}(\mathcal{L})$ for some lattice
	$\mathcal{L}\subset {\mathbb{Z}}^m$ having $\gamma_n = 2^{n/2}$ 
	and $n=rank(\mathcal{L}),$ in polynomial time. For  more details on Babai algorithm see also \cite[Section 2]{mar2}
	\subsubsection{SVP and NTRU}\label{SVPNTRU}
	Let the lattice $L_{\bf h}$ generated by the rows of the matrix
	$$M_{\bf h}=
		\left[\begin{array}{c|c}
			I_N & C({\bf h})  \\
			\hline
			{\bf 0}_N & qI_N   \\
		\end{array}\right],
			$$ 
where $C({\bf h})$ is the circulant matrix generated by the vector ${\bf h},$ see the definition in (\ref{formula:circulant}).  
			This matrix is public, since it contains the public key of the NTRU cryptosystem. From ${\bf f}(x)\star {\bf h}(x)\equiv {\bf g}(x)\pmod{q},$ there is a polynomial ${\bf b}(x)\in R$ such that $ {\bf f}(x)\star {\bf h}(x)- q{\bf b}(x)= {\bf g}(x),$ so considering polynomials as vectors we get ${\bf f}C({\bf h})-q{\bf b}={\bf g},$ thus $({\bf f},{-\bf b})M_{\bf h} = ({\bf f},{\bf g}).$ That is, $({\bf f},{\bf g})\in L_{\bf h}.$
	We can see that  
	\begin{equation}\label{definition_of_lattice}
	L_{\bf h} = \{ ({\bf u}(x),{\bf v}(x))\in R^2 : {\bf u}(x)\star {\bf h}(x) \equiv {\bf v}(x) \pmod{q}\}
	\end{equation}
	or 
	\[
	L_{\bf h} = \{ ({\bf u},{\bf v})\in {\mathbb{Z}}^{2N} : {\bf u}C({\bf h}) = {\bf v} \}.
	\]
	Thus, the problem of finding the private key $({\bf f},{\bf g})\in {\mathbb{Z}}_q^{2N},$ comes down to find a short vector in the integer lattice $L_{\bf h}.$ Note that, $L_{\bf h}$ is a $q-$ary lattice, i.e. $q{\mathbb{Z}}^{2N}\subset L_{\bf h}.$ 
	\subsubsection{CVP and NTRU}\label{CVPNTRU}
Closest Vector Problem (CVP) can be used to recover the message and not the private key as previous in subsection \ref{SVPNTRU}. To see how we apply CVP to recover the message ${\bf m},$ we first note that $(p{\bf r},{\bf e}-{\bf m})\in L_{\bf h}.$ Indeed, from the form of lattice $L_{\bf h}$, see (\ref{definition_of_lattice}), and from the encryption of ${\bf m}$, see (\ref{definition_of_encryption}), we get 
$p{\bf r}\star {\bf h} = {\bf e}-{\bf m}\pmod{q},$ therefore $(p{\bf r},{\bf e}-{\bf m})\in L_{\bf h}.$ 
Now, we write 
$$({\bf 0}_N,{\bf e}) = (p{\bf r} - p{\bf r}, (p{\bf r}\star {\bf h} +{\bf m}){\ \rm mod}{q})=$$ 
$$(p{\bf r},p{\bf r}\star {\bf h}{\ \rm mod}{q}) + (-p{\bf r},{\bf m}). $$
Finally, the vector $(p{\bf r},p{\bf r}\star {\bf h}{\ \rm mod}{q})=(p{\bf r},{\bf e} - {\bf m})\in L_{\bf h}$ and the vector 
$(-p{\bf r},m)$ is quite short. To see this and assuming that $p=3$ i.e. ${\bf r}, {\bf m}\in \{-1,0,1\}^N,$ we get
$${\rm{dist}}\big( (p {\bf r},{\bf e} - {\bf m}),({\bf 0},{\bf e})\big)=||(p{\bf r},{\bf e} - {\bf m})- ({\bf 0},{\bf e})||=||(p{\bf r},-{\bf m})||\leq \sqrt{10N}.$$ 
Therefore, in order to find ${\bf m},$ we apply CVP in $L_{\bf h}$ with target vector $({\bf 0}_N,{\bf e}).$ We generalize this method by using a modification of the target vector which is closer to the previous lattice, see Remark \ref{remark:CVP}. 
\ \\\\
	{\bf Roadmap.}
	In section \ref{sec:previous_work} we provide the bibliography and results that concerns attacks to NTRU and also we present our contribution.
	In section \ref{sec:auxiliary} we present some auxiliary results which we shall use in our attack in the next section \ref{attack_to_ntru}. In subsection \ref{subsection:the_attack} we provide the algorithm of the attack and we depict in some examples, using Sagemath and Fpylll, the success of our attack. In the final section, we end up by presenting some concluding remarks and commenting possible future work.

\section{Previous work-Our Contribution}\label{sec:previous_work}
	In 1997 Coppersmith and Shamir \cite{CopSha97} proposed a lattice based attack on the NTRU cryptosystem. Their attack uses lattice reduction to find the private key\footnote{This attack does not have any relation with the Coppersmith attack as used, for instance in RSA (\cite{cop-rsa}) or (EC)DSA (for instance see \cite{blake,draz-dsa})}. In \cite{gentry} the author proposed lattice attacks that are efficient, when $N$ is composite, by reducing lattices of small dimension to find partial information about the secret key. Furthermore, in \cite{May} May used a different class of lattices, which he called {\it {run-lattices}}. This class of lattices is constructed from the classic NTRU-lattice by multiplying columns $N+1$ till $N+r$ by a constant $\theta$.

	In \cite{Silverman}, Silverman  generalized May's idea and he proposed a method where he selects $r$ coefficients and then force them  equal to zero by reducing the dimension of the lattice.  In \cite{Gama-Nguyen} the authors present some new chosen-ciphertext attacks. The attacks exploit the decryption failures. 
	
	In \cite{Odlyzko}, the authors presented the meet-in-the middle attack, which was first observed by Andrew Odlyzko. The idea of this attack is to split the search space for the secret key into two parts and use a collision search algorithm. Thus, one reduces the steps by a square root. The drawback is that it needs a significant amount of memory. 
	
	In \cite{hybrid attack} Howgrave-Graham proposed the improved Hybrid Attack. This is a combination of lattice reduction and meet-in-the-middle techniques. He compared his attack with the meet-in-the-middle proposed by Odlyzko and he deduced that the algorithm of Hybrid Attack requires about $2^{60.3}$ loops in order to recover the private key of \texttt{ee251ep6} parameter set. On the other hand Odlyzko's attack requires about $2^{84.3}$ loops. What is more, it requires a factor of $2^{24.6}$ less storage. This is the most practical attack, but it also needs exponential time. The Hybrid Attack has been used to estimate the security of many lattice-based cryptographic schemes. In \cite{Buchmach, Wunderer} the authors support that the analysis in \cite{hybrid attack} as well as those in other schemes are not entirely satisfactory, leading to unreliable estimates. Specifically the authors in \cite{Buchmach} improved the analysis of the runtime of the attack by proving that in some cases, meet-in-the-middle attack is better than  Hybrid Attack. Furthermore, they proposed a generalized version of the hybrid attack for solving SVP and BDD problems in $q-$ary lattices. 
	In 2016, Albrecht, Bai and Ducas \cite{Albrecht} and independently Cheon, Jeong and Lee \cite{Cheon} proposed  much the same methods to attack the NTRU cryptosystem with larger modulus than in the NTRUEncrypt standard. The main idea is to decrease the dimension of the NTRU lattice using the multiplication matrix by the norm (resp. trace) of the public key in some subfield. In \cite{Kirchner} the authors  presented a new variant of the subfield attacks which is better than both of the two previous attacks in practice. They proved that in $\mathbb {Q}(\zeta _{2^n}),$ the time complexity is polynomial for  $q=2^{\mathrm {\Omega }(\sqrt{n\log \log n})}.$ Furthermore, they made a comparison between this attack and the hybrid attack, concluding that hybrid attack is better in practice.
	
	Finally, in 2021 Nguyen \cite{Nguyen - Boosting the hybrid} analyses the meet in the middle and the hybrid attack by making some simplifications and further improvements. What is more, he deduced that the security estimates of the NTRU finalist in NIST’s post-quantum standardization need to be revised.
	
	Another very interesting line of research presented in 2005 by Silverman, Smart and Vercaturen in \cite{witt} and extended in 2009 by  Bourgeois and Faugère in \cite{gerald}, where they used Witt vectors to reduce the NTRU problem to a multivariate quadratic system  over the Galois field with two elements.
	\subsection{Our contribution}
	There is the classic message recovery attack using closest vector problem, see subsection \ref{CVPNTRU}. In the present paper we generalize the classic CVP attack. In fact we apply a CVP attack to recover the message ${\bf m}$ without using the NTRU lattice based on the public key. We remark here that there not many results in this direction, i.e. attacks based on CVP. On the other hand, there are plenty of such results, i.e. based on CVP, for other cryptosystems, for instance (EC)DSA (Digital Signatur Algorithm). To our knowledge, in the bibliography, there are not results, concerning attacks to NTRU that are based on CVP on some lattice, except the classic that we have already
presented in subsection \ref{CVPNTRU}.
	
In fact we apply CVP to a lattice ${\mathcal{L}}_{\bf a}$ (instead of $L_{\bf h}$) for some fixed and suitably chosen vector ${\bf a}\in{\mathbb{Z}}^{2N}$ and target vector $({\bf 0}_N,{\bf a}\star {\bf e} + {\bf E})$ (instead of (${\bf 0}_N,{\bf e})$ in the classic CVP attack), for some suitable vector ${\bf E}$. The new lattice does not depend on the public polynomial ${\bf h}(x)$ and depends only on $N, q,$ and a (real) parameter ${y}.$ The new idea here is that, this new lattice ${\mathcal{L}}_{\bf a}$ allows us to provide a deterministic attack to NTRU under the assumption that we know an approximation of the unknown vector ${\bf E}$ and assuming that we have a CVP oracle. 

Furthermore, in practice we implement the attack using Babai's nearest plane algorithm on the lattice ${\mathcal{L}}_{\bf a}$, where the vector ${\bf a}$ depends on $N, q,$ and a fixed real parameter $y.$  The target vector, as we shall see, is a sum of two vectors. The first vector is known, but the second vector needs some guesses on the part of the attacker. The unknown part of target vector is a multiple of the nonce ${\bf r}$. If, for instance, we have a weak generator for ${\bf r}$, we might predict some digits of ternary polynomial ${\bf r}$ and so we get a better estimate for the target vector. In our examples, the previous approximation is provided by a suitable oracle. It may seem unrealistic to have such an oracle in practice, but a weak random generator or an implementation of a side channel attack, it could provide us such an oracle.
 
Finally, since the lattice does not depend on the public key ${\bf h}$, we can use a LLL/BKZ reduction only one time, as far as $(N,q,y)$ remain the same. So, if the keys changed, we do not have to repeat the reduction step of the closest vector problem. 

	 As far as we know this attack is new. We shall present some experiments, where we show that our attack is successful. In fact, in subsection \ref{examples}\footnote{For the code see \url{https://github.com/drazioti/ntru}}, in examples 5 and 6, we consider the state of the art parameters : ntruhps2048509 and ntruhps2048677 i.e. $(N,q)=(509,2048)$ and $(677,2048),$ respectively.
	 These parameters were recommended in the 3rd NTRU submission in NIST \cite{ntru-third}.   		
	\section{Auxiliary Results} \label{sec:auxiliary}
	Let $N>2$ be a prime number and ${\bf a}=(a_0,\dots,a_{N-1})$ be a random vector of $R_q.$ With ${\mathcal{L}}_{\bf a}$ we denote the lattice generated by the rows of the matrix,	
	
	\begin{equation}\label{definition_of_matrix}
		M_{\bf a}=
		\left[\begin{array}{cccc|cccc}
			1 &     0 & \dots &   0      &   a_0    &  a_1   & \dots & a_{N-1}  \\
			0&     1 & \dots &   0       &   a_{N-1}&  a_0   & \dots & a_{N-2}  \\
			\vdots & \vdots  &  \ddots & \vdots &   \vdots & \vdots & \ddots & \vdots   \\
			0 & 0 & \dots & 1            &   a_1    &   a_2  &  \dots &   a_0  \\
			\hline
			0 & 0 &  \dots &    0        &     q     & 0     &  \dots &    0    \\
			0 & 0 &  \dots &    0        &     0     & q   &  \dots &    0    \\
			\vdots & \vdots  &  \ddots & \vdots &   \vdots & \vdots & \ddots & \vdots   \\
			0 & 0 &  \dots &    0        &     0     & 0   &  \dots &    q    \\
		\end{array}\right]
		\end{equation}
	For ${\mathcal{L}}_{\bf a}$ we make the following assumption :\\
	{\bf Assumption}. Let $q$ be a positive integer and $y\geq 1$ be a real number. We assume that 
	\begin{equation}\label{assumption}
	\lambda_1({\mathcal{L}}_{\bf a}) > q^{1/y}.
	\end{equation}
	Since it is not possible to compute $\lambda_1$ for large values of $N,$ we can not really check the validity of this inequality. We shall use Gaussian heuristic to get a more practical form of this inequality.
		The Gaussian heuristic $GH(\mathcal{L}_{\bf a})$ for ${L}_{\bf a}$ suggests that the length of a shortest vector is approximately, 
		$$\sqrt{\frac{2N}{2 \pi e}}\det(M_{{\bf a}})^{1/2N}=\sqrt{\frac{qN}{\pi e}}\approx 0.35\times \sqrt{qN}.$$ 
		However, according to \cite[Proposition 6.61]{hoffstein} the private key is about $\frac{1}{\sqrt{N}}$ the bound suggested by Gaussian heuristic. So, we update the previous bound $:$
		$$\lambda_1\approx\frac{1}{\sqrt{N}}0.35\times \sqrt{qN} = 0.35\sqrt{q}.$$ Now, under the previous heuristic, the assumption is written,  
\begin{equation}\label{heuristic_inequality}
0.35\sqrt{q}>q^{1/y}.
\end{equation}		
			It is easy to check that there are plenty of $(N,q,y)$ that satisfy the two previous inequalities. For instance, see Appendix B, for some examples where the heuristic inequality (\ref{heuristic_inequality}) and the inequality (\ref{assumption}) are both valid.

	The following Proposition will allows us to implement an attack on NTRU. 
	\begin{proposition}\label{prop1}
	Let
		$${\bf a}(x) = a_0 + a_1x +\cdots + a_{N-1}x^{N-1}$$ be a polynomial in $R_q.$
		 Consider the equation in $R_q$,
		$${\bf a}(x) \star {\bf m}(x) = {\bf b}(x) + {\bf c}(x),$$
		where the unknowns are the polynomials
		$$ {\bf m}(x) = \sum_{j=0}^{N-1}m_jx^{j},\  \  {\bf c}(x) = \sum_{j=0}^{N-1}c_jx^j,$$
		and ${\bf b}(x) =  b_0+b_1x+\cdots + b_{N-1}x^{N-1}$  with
		$b_0,\ldots,b_{N-1}\in\{0,\ldots,q-1\}$.
		Let   ${\bf V}=({\bf m},{\bf c})$ be a vector
		with coordinates the coefficients of  a solution $({\bf m}(x), {\bf c}(x)).$ If we can find  a vector $ {\bf E}=(E_0,E_1,\ldots,E_{N-1},E_N,\ldots,E_{2N-1})  \in \mathbb{Z}^{2N}$  satisfying
		$$\|{\bf V}-{\bf E}\|<\frac{1}{2}q^{\frac{1}{y}},$$
		then we can determine the solution vector ${\bf V}$ using a CVP oracle.
	\end{proposition}
	
	\begin{proof}	
		Let ${\bf u}=(m_0,m_1,\ldots,m_{N-1},b_0 + c_0,\ldots,b_{N-1} + c_{N-1})$.
		We construct the lattice ${\mathcal{L}}_{\bf a}$ spanned by the rows of the $2N\times 2N$  matrix $:$ \\
		$$
		M_{\bf a}= 
		\left[\begin{array}{c|c}
			I_N & C({\bf a})  \\
			\hline
			{\bf 0}_N & qI_N   \\
		\end{array}\right]
		$$ 
		It is easy to see that ${\bf u}$ is a lattice point. Indeed, 
		there is a polynomial ${\bf v}(x)\in R$ such that,
		$${\bf a}(x)\star {\bf m}(x) = {\bf b}(x) + {\bf c}(x) + q{\bf v}(x).$$
		Then, 
		$$({\bf m},-{\bf v})M_{\bf a}=({\bf m},-{\bf v})\left[\begin{array}{c|c}
			I_N & C({\bf a})  \\
			\hline
			{\bf 0}_N & qI_N   \\
		\end{array}\right]=({\bf m},{\bf m}C({\bf a}) - q{\bf v})=$$
		$$({\bf m},{\bf m}\star {\bf a}-q{\bf v})=({\bf m},{\bf a}\star {\bf m}-q{\bf v})=({\bf m},{\bf b}+{\bf c})={\bf u}.$$
		Let 
		$$ {\bf b}_{target} =  (E_0,E_1,\ldots,E_{N-1},b_0 + E_N,\ldots, b_{N-1} + E_{2N-1}).$$
		Remark that,
		$${\bf b}_{target} = {\bf E} + ({\bf 0}_N,{\bf b}).$$ 
		So,
		$${\bf b}_{target} + {\bf V} = {\bf E} + ({\bf 0}_N,{\bf b}) + {\bf V} 
		= {\bf E} + ({\bf m},{\bf b}+{\bf c}) = 
		{\bf E}+{\bf u}.$$
		Therefore,
		$$	\|{\bf u}- {\bf b}_{target}\| = \|{\bf V}- {\bf E} \| < \frac{1}{2}\, q^{\frac{1}{y}} .$$
		We call a CVP Oracle with input the lattice ${\mathcal{L}}_{\bf a}$ and target vector 
		$ {\bf b}_{target}$ and we get a vector ${\bf w} \in {\mathcal{L}}_{\bf a}$ such that,
		$$\|{\bf w}- {\bf b}_{target}\| \leq	\|{\bf u}- {\bf b}_{target}\| = \|{\bf V}- {\bf E} \| < \frac{1}{2}\, q^{\frac{1}{y}} .$$
		So we get, 
		$$\|{\bf w}- {\bf u}\| \leq	\|{\bf w}- {\bf b}_{target}\| + \|{\bf u}- {\bf b}_{target}\|  < \frac{1}{2}\, q^{\frac{1}{y}} + \frac{1}{2}\, q^{\frac{1}{y}} = q^{\frac{1}{y}} .$$ Our assumption (\ref{assumption}) implies that every lattice vector of ${\mathcal{L}}_{\bf a}$ has length at least $q^{1/y}.$ Since ${\bf u}$ and ${\bf v}$ are lattice vectors, then also ${\bf u}-{\bf v}$ is a lattice vector, so $||{\bf u}-{\bf v}||>q^{1/y}.$ Therefore,  $ {\bf w}- {\bf u} = {\bf 0} $ or $ {\bf w} = {\bf u}$. The result follows.
	\end{proof}
	\begin{remark}\label{remark:CVP}
	We note that, if we take the original lattice $L_{\bf h},$ we can have a similar Proposition for the classic CVP (see \ref{SVPNTRU}), taking as target vector $({\bf 0}_N,{\bf e})+{\bf E}.$ So, the Proposition generalizes the classic CVP attack. 
	\end{remark}
	\section{The Attack to NTRU}\label{attack_to_ntru}
	
	In this section we show how we can find the plaintext ${\bf m}(x)$ using the previous Proposition \ref{prop1}. 
	In this Proposition we provide only a sufficient condition. It may happen to find the message, even if the inequality is not being satisfied.		
	The encrypted message ${\bf e}(x)$ is given by the following relation,
	$${\bf e}(x) \equiv p{\bf r}(x) \star{\bf h}(x) + {\bf m}(x) \, \bmod\,  q.$$
	By multiplying both sides with ${\bf a}(x)$ (where ${\bf a}(x)$ is a random polynomial of ${\mathcal{R}}_q$) we get,
	$${\bf a}(x)*{\bf m}(x) \equiv{\bf a}(x)*{\bf e}(x) - p{\bf a}(x)*{\bf r}(x)*{\bf h}(x)  \, \bmod\,  q $$
	or
	\begin{equation}\label{above_equation}
		{\bf a}(x)*{\bf m}(x) \equiv {\bf b}(x) + {\bf c}(x)  \, \bmod\,  q.
	\end{equation}
	We have set
	$${\bf b}(x)= {\bf a}(x)*{\bf e}(x) \mod{q} $$
and
	 $$ {\bf c}(x)= - p{\bf a}(x)*{\bf r}(x)*{\bf h}(x) \mod{q}.$$
	In equation (\ref{above_equation}), ${\bf a}(x)$ and ${\bf b}(x)$ are known.
	We assume that, the following vector
	\begin{equation}\label{vector_V}
		{\bf V}=(m_0,m_1,\ldots,m_{N-1},c_0,c_1\ldots,c_{N-1})
	\end{equation} 
	is a solution of equivalence (\ref{above_equation}).
	We construct the lattice $\mathcal{L}_{\bf a}$ spanned by the rows of the $2N\times 2N$  matrix  $
	M_{\bf a}$ as in (\ref{definition_of_matrix}).
	Let    $ {\bf E}=(E_0,E_1,\ldots,E_{N-1},E_N,\ldots,E_{2N-1})  \in \mathbb{Z}^{2N}$  such that,
	$$\|{\bf V}-{\bf E}\|<\frac{1}{2}q^{\frac{1}{y}}. $$
	We call our CVP Oracle with input the lattice ${\mathcal{L}}_{\bf a}$ and target vector 
	$$ {\bf b}_{target} = (E_0,E_1,\ldots,E_{N-1},b_0 + E_N,\ldots, b_{N-1} + E_{2N-1}).$$ Say that we get a vector ${\bf w} \in {\mathcal{L}}_{\bf a}$. Then, Proposition \ref{prop1} yields 
	$$w_0 = m_0, \ldots ,w_{N-1 }= m_{N-1}.$$ In this way we compute the message   ${\bf m}(x)$.
	
	\subsection{The Attack}\label{subsection:the_attack}	
	In this subsection, we present our message recovery attack based on the previous analysis.
	
	\ \\
	{\bf Algorithm 1}\\
	{\texttt{INPUT:  $N:$ prime, $y\in{\mathbb{R}}_{>1},$
			${\bf e}$ the encryption of a message ${\bf m}$, ${\bf E}\in {\mathbb{Z}}^{2N},$ and $q$ which is a prime or a prime power.}} \\
	{\texttt{OUTPUT: a message {${\bf m}'$}}}\\\\
	\texttt{1.}   ${\bf a}\xleftarrow{\$} 
	\{0,1\}^{N-1}\times \{\lfloor Nq^{1/y}\}\rfloor$\\
	\texttt{2.}  Construct the matrix $M_{\bf a}$ and the lattice $\mathcal{L}_{\bf a}$ generated by the rows of $M_{\bf a}$\\
	\texttt{3.}    ${\bf b}\leftarrow {\bf a}\star {\bf e}$ in $R_q$ \\
	\texttt{4.} ${\bf b}_{target}\leftarrow ({\bf 0}_N;{\bf b})+{\bf E}$\\
	\texttt{5.} ${\bf w}\leftarrow {\tt babai}({\mathcal{L}}_{\bf a},{\bf b}_{target})$\ \\
	\texttt{6.} ${\bf m'}\leftarrow$ ${\bf w}[0:N-1]$\ \# the $N-$first coordinates of ${\bf w}$ \\
	\texttt{7.} {\bf return} ${\bf m}'$ \\\\	
	Since NTRUEncrypt is used in Key Encapsulation Mechanisms, with the previous attack we expect to find the shared key. Indeed, if the message is right, then we can compute the nonce ${\bf r}$ and then taking a suitable hash we get the shared key. The validity of the key can only be checked indirectly, by starting some encrypted conversation with the owner of the secret key, by using a symmetric cryptosystem.
	 
	In line 1, we choose small entries for the vector ${\bf a}$ except one which is near to $Nq^{1/y}.$ Someone would expect to pick ${\bf a}$ randomly from $R_{q}.$ There are some theoretic arguments for this choice. We explain this in Appendix A.
	In line $6,$ we return the first $N$ coordinates of ${\bf w},$ i.e. $w_0,...,w_{N-1}.$ 
	In line $5$ before applying Babai's algorithm, we reduce the basis by using LLL. The LLL reduction is the same for all $(N,q,y),$ so we can make this step independently from the others, i.e. we can apply LLL in line $2.$
	
	We tried the previous attack for various parameters.
	In all the examples we assume that ${\bf E} = ( {\bf r}_N, {\bf E}'  ),$
	where ${\bf r}_N\in\{-1,0,1\}^N,$ ${\bf E}' = (E_1',....,E_N')$ with $|E_i'- c_i|\leq R,$ 
	and $c_i$ are the coefficients of ${\bf c}(x)$. Note that, ${\bf c}(x)$ is the product in $R_q$ of $-p\star {\bf h}(x)\star {\bf a}(x)$  and the unknown nonce ${\bf r}(x),$ which is a short ternary polynomial. Furthermore, whenever the attack succeeded, we recovered the message ${\bf m},$ in less than a minute. For the implementation we used Sagemath \cite{sage} and for the LLL reduction and Babai algorithm, we used fpylll library \cite{fpylll}. Before we present our 
	experiments, we describe our attack by using an oracle. 
	
	Say we have an oracle ${\mathcal{O}}$ that on input the public key of the system, the parameter $R,$ and a seed, provides 
	us with some integers $E_i', $ $i=1,2,...,N,$ such that $|E_i'- c_i|\leq R.$ 
	Each call to the oracle with a different seed provides a new random set of $E_i'$ with the same property. We assume that the distribution of $E_i'$ with this property is uniform.
	\subsubsection{Examples}\label{examples}
	In our examples, ${\mathcal{L}}_f={\mathcal{T}}(d+1,d)$ and ${\mathcal{L}}_g={\mathcal{L}}_r={\mathcal{T}}(d,d).$ In the document submitted to the third phase of NIST \cite{ntru-third}, the authors recommend for the case of NTRU-HPS, ${\mathcal{L}}_f$ be the set of ternary  polynomials (i.e. with coefficients in $\{-1,0,1\}$) of degree at most $N-2.$ In examples 1-4, we follow Algorithm 1 for the choice of vector ${\bf a}.$ All the examples use LLL and Babai's algorithm, thus are very fast.\\
	{\bf Example 1.} $N = 239, d = 71, p = 3, q=2^8=256, y=2.3, {\bf r}_N={\bf 0}$.\\
	In  this experiment for $R=9$ we got the message after some calls ($<100$). For $R=10$ we did not manage to find the message.
	\ \\
	{\bf Example 2.} $N = 257, d = 91, p = 3, q=2^8=256, y=2.3, {\bf r}_N={\bf 0}$.\\
	In  this experiment for $R=9$ after some calls ($<100$) we got the message. For $R=10$ we did not manage to find the message after 100 calls to the oracle.
	\ \\
	{\bf Example 3.} $N = 283, d = 99, p = 3, q=2^{10}=1024, y=2.3, {\bf r}_N={\bf 0}$.\\
	In  this experiment for $R=16$ after some calls ($<100$) we got the message. For $R=17$ we did not manage to find the message.
	\ \\
	{\bf Example 4.} $N = 307, d = 15, p = 3, q=2^{10}=1024, y=2.5, {\bf r}_N={\bf 0}$.\\
	In  this experiment for $R=18$ we got the message after at most 100 calls to the oracle. For $R=19$ we did not manage to find the message.\ \\
	{\bf Example 5.} $N = 509, d = 10, p = 3, q= 2^{11}, y=2.5, {\bf r}_N={\bf 0}$.\\
For this experiment we picked ${\bf a}$ from $\{-2,2\}^{N-1}\times \{ \lfloor Nq^{1/y} \rfloor\}$	and we randomly shuffle it. In  this experiment for $R=26$ after at most 100 calls we got the message. For $R=27$ we did not manage to find the message after 100 calls. The specific parameters were suggested in \cite{ntru-third} for ntruhps2048509.
\ \\\\
For the next examples we used the following,
$${\bf a}=(-k,-k+1,...,-1,1,2,...,k,\lfloor Nq^{1/y} \rfloor +1),$$
where $k=\frac{N-1}{2}.$\ \\
	{\bf Example 6.} $N = 677, d = 20, p = 3, q= 2^{11}, y=2.5, {\bf r}_N={\bf 0}$.\\ 
	The authors of \cite{ntru-third} for better security suggested	ntruhps2048677. For $R=17$ we got almost immediately the message, assuming that we have the LLL-reduced basis. The LLL reduction took about 15 minutes in Fpylll \cite{fpylll}.  For $R=18$ after 100 calls to the oracle we did not find the message.\ \\
{\bf Example 7.} $N = 557, d = 40, p = 3, q= 2^{13}, y=2.5, {\bf r}_N={\bf 0}$.\\
In  this experiment for $R=38$ after at most 100 calls we got the message. For $R=39$ we did not manage to find the message after 100 calls. The parameters $(N,q)=(557,2^{13})$ were recommended in \cite{HRSS}.
\begin{remark}
In the previous examples the heuristic inequality (\ref{heuristic_inequality}) is not satisfied (but the two sides of the inequality are very close). Proposition \ref{prop1} is not if and only if statement, i.e. if assumption (\ref{assumption}) is satisfied then we get the result of the Proposition. However, it may occur, the assumption be false and the attack may still works.
\end{remark}
	\begin{remark}
		The experiments suggest that if we increase $q=2^e,$ greater values for the range $R$ are allowed in order to get a successful attack. 
	\end{remark}
	\begin{remark}
		Assume that ${\bf r}_N={\bf 0},\ p=3$ and for some $E_i'$ we have $|E_i'-c_i|<R$ $(1\leq i\leq N).$ We consider the notation of section \ref{attack_to_ntru}. Then,
		$$||{\bf V}-{\bf E}||^2=||{\bf m}||^2 + \sum_{i=1}^{N}(E_i'-c_i)^2\leq N+R^2N=N(1+R^2).$$
		Also if $N(1+R^2)<\frac{1}{4}q^{2/y}$ or 
		$$y<\frac{2\log_2{q}}{2+\log_2(N(1+R^2))},$$
		then from Proposition \ref{prop1} we get that a CVP oracle will provide us the solution vector ${\bf V}.$	\end{remark}
	\section{Conclusion}\label{sec:conclusion}
	In this work we used lattice theory to attack the NTRU-HPS cryptosystem. Multiplying the encryption equation with a (random) polynomial, we create an equivalent equation, where an unknown value belongs to a  lattice, for which we can find a lower bound of the first successive minima of the lattice. Then, we choose a suitable target vector and apply Babai's nearest plane algorithm with the hope that the output  is the unknown value. The difficult part is the choice of vector ${\bf E}$, thus we need the help of a suitable oracle. This is a drawback of the attack, since it is difficult in practice to have such an oracle. To address this problem someone has to apply a side channel attack, hoping to get some information about the vector ${\bf c}$ of Proposition \ref{prop1}. Then, the attacker may have a good guess for the vector ${\bf E}.$  For instance, side channel attacks were studied in \cite{aydin,kamal,vizev}. Finally, we provided several examples showing the success of the attack with the use of the oracles.

{\bf Acknowledgment}.
\begin{figure}[!htb]
\centering
\includegraphics[scale=0.5]{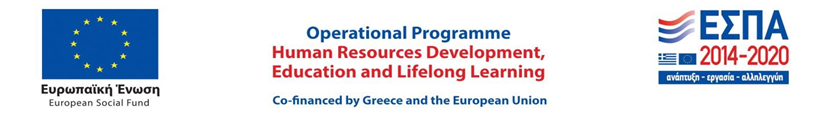}
\end{figure}
Marios Adamoudis is co-financed by Greece and the European Union (European Social Fund-ESF) through the Operational Programme "Human Resources Development, Education and Lifelong Learning" in the context of the Act "Enhancing Human Resources Research Potential by undertaking a Doctoral Research" Sub-action 2: IKY Scholarship Programme for PhD candidates in the Greek Universities.

Finally, the authors sincerely thank professor Poulakis for his helpful suggestions.

\appendix

\section{How to pick ${\bf a}$?}\label{appendixA}

\begin{proposition}\label{oldTheorem1}
		Let $q$ and $N$  be positive integers with  $N=2k+1,$ and $y$ a real number $ \ge 1.$ 		
		We set, 
		\[ 
		a_j(N,q,y)=a_j= \left\{
		\begin{array}{ll}
			j-k, & j=0,1,...,k-1 \\
			j+1-k, & j=k,k+1,...,2k-1=N-2\\
			\lfloor Nq^{1/y} \rfloor + 1, & j=2k=N-1 \\
		\end{array} 
		\right. 
		\]	
		For instance if $N=11, q=2^5, y=1$ we get 
		$${\bf a} = (-5,-4,-3,-2,-1,1,2,3,4,5,353).$$
		We assume that,
		\begin{equation}\label{second_inequality}
			a_{0}^2 + a_{1}^2 + \dots +a_{N-2}^2+ a_{N-1}^2 =\frac{(N^2-1)N}{12} + a_{N-1}^2 < \frac{1}{N} \frac{(q-q^{1/y})^2}{q^{2/y}}
		\end{equation}
		Let  ${\bf a}=(a_j)_j$ and $\mathcal{L}_{\bf a}={\mathcal{L}}(M_{\bf a})$ be the lattice spanned by the rows of the $2N\times 2N$ matrix  
	$$
		M_{\bf a}= 
		\left[\begin{array}{c|c}
			I_N & C({\bf a})  \\
			\hline
			{\bf 0}_N & qI_N   \\
		\end{array}\right]
		$$ 	
		Then, for every ${\bf v}\in \mathcal{L}_{\bf a} - \{{\bf 0}\}$ not belonging to the lattice generated by the rows of the matrix
		$$ 
		\left[\begin{array}{c|c}
			I_N & C({\bf a})  
		\end{array}\right]
		,$$  we have$:$
		$$\|{\bf v}\|> q^{\frac{1}{y}}.$$		
	\end{proposition}
	
	\begin{proof}
		Assume that there is a non-zero vector ${\bf v}\in \mathcal{L}_{\bf a}$ such that, 
		$$\|{\bf v}\|\leq q^{\frac{1}{y}}.$$ Let ${\bf b}_1,\dots,{\bf b}_{2N}$ be the rows of the matrix $M_{\bf a}.$ We define the following symbol, for $a$ integer and $n$ positive integer,
		\[ 
		[a]_n= \left\{
		\begin{array}{ll}
			a+n, & a<0<n \\
			a, & 0\le a<n\\
		\end{array} 
		\right. 
		\]	
		Since ${\bf v}\in \mathcal{L}_{\bf a},$ there are integers $l_1,\ldots,l_{2N}$  such that,
		$${\bf v}=l_1{\bf b}_1+\cdots+l_{2N}{\bf b}_{2N} = $$
		$$\Big(l_1,l_2,\dots,l_{N},\sum_{j=1}^N l_ja_{[1-j]_{N}}+ql_{N+1},\sum_{j=1}^N l_ja_{[2-j]_{N}}+ql_{N+2},\dots, \sum_{j=1}^N l_ja_{[N-j]_N}+ql_{2N}\Big).$$
		Then we get
		\begin{equation}\label{second} 
			\begin{dcases}
				|l_1|,|l_2|,\dots,|l_N| \le q^{\frac{1}{y}}\\
				|l_1a_0+l_2a_{N-1}+\dots+ l_Na_{1} + ql_{N+1}| \le q^{\frac{1}{y}}\\
				|l_1a_{1}+l_2a_0+\dots+l_Na_{2}+ql_{N+2}| \le q^{\frac{1}{y}}\\  
				\dots\\
				|l_1a_{N-1}+l_2a_{N-2}+\dots+l_Na_0+ql_{2N}| \le q^{\frac{1}{y}}
			\end{dcases}
		\end{equation}	
		Let $(i_1,...,i_N)$ be a random right rotation of $(0,1,...,N-1).$
		By the Cauchy-Schwarz inequality and relations (\ref{second_inequality}) and (\ref{second}) we 
		obtain: 
		$$|l_1a_{i_1}+l_2a_{i_2}+\dots+ l_Na_{i_n}|^2 \le (l_1^2 + l_2^2 +\dots+l_N^2 ) (a_0^2 + a_1^2 +\dots+a_{N-1}^2 )\le  \  \ \ \  \ \ \ \ \ \ \ \ \ \ $$
		$$( q^{\frac{2}{y}}+q^{\frac{2}{y}} + \dots +q^{\frac{2}{y}} )(a_0^2 + a_1^2 +\dots+a_{N-1}^2 )= Nq^{\frac{2}{y}} (a_0^2 + a_1^2 +\dots+a_{N-1}^2 ) < (q-q^{\frac{1}{y}})^2.$$
		Therefore, we have
		\begin{equation}\label{third} 
			|l_1a_{i_1}+l_2a_{i_2}+\dots+ l_Na_{i_n}|< q-q^{\frac{1}{y}}.
		\end{equation}
		Since ${\bf v}$ does not belong to the lattice generated by the rows of the matrix,
		$$ 
		\left[\begin{array}{c|c}
			I_N & C({\bf a})  
		\end{array}\right]
		,$$  
not all the integers $l_{N+1},l_{N+2},\dots,l_{2N}$ will 
be zero.
		Say $l_{N+1} \neq 0.$ Thus, we get
		\begin{eqnarray*}
			\lefteqn{\|{\bf v}\| \ge |l_1a_0 + l_2a_{N-1}+\dots+ l_Na_{1} +l_{N+1}q|
				\ge }\\
			& & |l_{N+1 }|q-|l_1a_0+l_2a_{N-1}+\dots+ l_Na_{1}| \ge \\ & & q -|l_1a_0+l_2a_{N-1}+\dots+ l_Na_{1}| > q^{\frac{1}{y}}, 
		\end{eqnarray*}
		which is a contradiction. The Proposition follows.
	\end{proof}
	This Proposition provide us some constraints on how to choose ${\bf a}.$ The advantage here is that, we do not need the inequality (\ref{assumption}), which is our assumption. Inequality (\ref{second_inequality}) can be checked easily for every $N,q,$ and $y,$ whereas our assumption is hard to be checked for large values of $N.$
\section{Some examples that satisfy the inequalities (\ref{assumption}) and (\ref{heuristic_inequality})}
	To produce instances\footnote{For the code see \url{https://github.com/drazioti/ntru/blob/main/appendix.ipynb}} that satisfy the two inequalities we need to execute exact SVP to compute $\lambda_1.$ We use moderate values of $N$ and the SVP function of Fpylll. The heuristic inequality (\ref{heuristic_inequality}) may be true, whereas the assumption (\ref{assumption}) may be false.  Usually, for large enough values of $y$ both inequalities are satisfied. For instance, for 
	$$(N,q,y) = (21,2^5,1.5) , (23,2^5,2) , (25,2^9,2), (27,2^9,2)$$
	and for randomly chosen vector ${\bf a},$ both the inequalities are satisfied.
	\end{document}